\documentclass[11pt]{scrartcl}
\usepackage{amsmath}
\usepackage{amssymb}
\usepackage{amsthm}
\usepackage{tikz}
\usepackage{algorithm}
\usepackage{algpseudocode}
\usepackage{pifont}
\usetikzlibrary{arrows}
\allowdisplaybreaks

\newtheorem{theorem}{Theorem}
\newtheorem{lemma}{Lemma}
\newtheorem{proposition}{Proposition}
\newtheorem{corollary}{Corollary}
\newtheorem{definition}{Definition}
\DeclareMathOperator*{\argmax}{arg\,max}
\DeclareMathOperator*{\E}{\mathbb{E}}

\title{Asymptotic Existence of Fair Divisions for Groups}
\author{
Pasin Manurangsi\\UC Berkeley
\and
Warut Suksompong\\Stanford University
}

\begin{document}

\date{}
\maketitle

\begin{abstract}
The problem of dividing resources fairly occurs in many practical situations and is therefore an important topic of study in economics. In this paper, we investigate envy-free divisions in the setting where there are multiple players in each interested party. While all players in a party share the same set of resources, each player has her own preferences. Under additive valuations drawn randomly from probability distributions, we show that when all groups contain an equal number of players, a welfare-maximizing allocation is likely to be envy-free if the number of items exceeds the total number of players by a logarithmic factor. On the other hand, an envy-free allocation is unlikely to exist if the number of items is less than the total number of players. In addition, we show that a simple truthful mechanism, namely the random assignment mechanism, yields an allocation that satisfies the weaker notion of approximate envy-freeness with high probability.
\end{abstract}

\section{Introduction}

Dividing resources among interested parties in a fair manner is a problem that commonly occurs in real-world situations and is consequently of fundamental importance. Countries negotiate over international issues, as Egypt and Israel did in 1978 over interests in the Sinai Peninsula \cite{BramsTo96} and the U.S. and Panama in 1994 over those in the Panama Canal. Likewise, divorced couples negotiate over their marital property, airlines over flight routes, and Internet clients over bandwidth and storage space. On a smaller scale, typical everyday tasks involving fair division include distributing household tasks, splitting a taxi fare, and sharing apartment rent \cite{GoldmanPr14}. Given its far-reaching and often critical applications, it should not come as a surprise that fair division has long been a popular topic of study in economics \cite{DubinsSp61,Foley67,Steinhaus48,Stromquist80,Varian74}.

To reason about fair division, we must carefully define what we mean for a division to be ``fair''. Several notions of fairness have been proposed in the literature. For example, a division is said to be \emph{proportional} if every player values her allocation at least $1/n$ times her value for the whole set of items, where $n$ denotes the number of players. Another commonly-used notion of fairness, and one that we will work with throughout the paper, is that of \emph{envy-freeness}. A division of a set of items is said to be \emph{envy-free} if each player values the set of items that she receives at least as much as the set of items that any other player receives. When utilities are additive, an envy-free division is proportional and moreover satisfies another important fairness notion called the \emph{maximin share criterion}.\footnote{See, e.g., \cite{BouveretLe14,Budish11} for the definition of the maximin share criterion.}
While procedures for finding envy-free divisions have been proposed \cite{BramsKiKl12,BramsKiKl14,HaakeRaSu02}, an envy-free division does not necessarily exist in arbitrary settings. This can most easily be seen in the case where there are two players and one item that both players value positively, or more generally when the number of players exceeds the number of items and players have positive values for items.

The fair division literature has for the most part assumed that each interested party consists of a single player (or equivalently, several players represented by a single common preference). However, this assumption is too restrictive for many practical situations. Indeed, an outcome of a negotiation between countries may have to be approved by members of the cabinets of each country. Since valuations of outcomes are usually subjective, one can easily imagine a situation in which one member of a cabinet of a country thinks that the division is fair while another member disagrees. Similarly, in a divorce case, different members of the family on the husband side and the wife side may have varying opinions on a proposed settlement. Another example is a large company or university that needs to divide its resources among competing groups of agents (e.g., departments in a university), since the agents in each group can have misaligned interests. Indeed, the professors who perform theoretical research may prefer more whiteboards and open space in the department building, while those who engage in experimental work are more likely to prefer laboratories. 

In this paper, we study envy-free divisions when there are multiple players in each group. Every player has her own preferences, and players in the same group can have very different preferences. In this generalized setting, we consider a division to be envy-free if every player values the set of items assigned to her group at least as much as that assigned to any other group.

\subsection{Our Contributions}

In Section \ref{sec:asymptotic}, we investigate the asymptotic existence and non-existence of envy-free divisions using a probabilistic model, previously used in the setting with one player per group \cite{DickersonGoKa14}. We show that under additive valuations and other mild technical conditions, when all groups contain an equal number of players, an envy-free division is likely to exist if the number of goods exceeds the total number of players by a logarithmic factor, no matter whether the players are distributed into several groups of small size or few groups of large size (Theorem \ref{thm:existence}). In particular, any allocation that maximizes social welfare is likely to be envy-free. In addition, when there are two groups with possibly unequal numbers of players and the distribution on the valuation of each item is symmetric, an envy-free division is likely to exist if the number of goods exceeds the total number of players by a logarithmic factor as well (Theorem \ref{thm:existencesymmetric}). Although it might not be surprising that a welfare-maximizing allocation is envy-free with high probability when there are sufficiently many items, we find the fact that only an extra logarithmic factor is required to be rather unexpected. Indeed, as the number of players in each group increases, it seems as though the independence between the preferences of each player would make it much harder to satisfy all of them simultaneously, since they all need to be allocated the same items. 

To complement our existence results, we show on the other hand that we cannot get away with a much lower number of items and still have an envy-free division with high probability. In particular, if the number of items is less than the total number of players by a superconstant factor, or if the number of items is less than the total number of players and the number of groups is large, we show that the probability that an envy-free division exists is low (Corollaries~\ref{cor:nonexistence1} and \ref{cor:nonexistence2}). This leaves the gap between asymptotic existence and non-existence of envy-free divisions at a mere logarithmic factor. 

While the techniques used to show asymptotic existence of envy-free divisions in Section \ref{sec:asymptotic} give rise to mechanisms that compute such divisions with high probability, these mechanisms are unfortunately not truthful. In other words, implementing these mechanisms in the presence of strategic players can lead to undesirable outcomes. In Section \ref{sec:existenceapprox}, we tackle the issue of truthfulness and show that a simple truthful mechanism, namely the random assignment mechanism, is \emph{$\alpha$-approximate envy-free} with high probability for any constant $\alpha\in[0,1)$ (Theorem \ref{thm:existenceapprox}). Approximate envy-freeness means that even though a player may envy another player in the resulting division, the values of the player for her own allocation and for the other player's allocation differ by no more than a multiplicative factor of $\alpha$. In other words, the player's envy is relatively small compared to her value for her own allocation. The number of items required to obtain approximate envy-freeness with high probability increases as we increase $\alpha$. Our result shows that it is possible to achieve truthfulness and approximate envy-freeness simultaneously in a wide range of random instances, and improves upon the previous result for the setting with one player per group \cite{AmanatidisBM16} in several ways.

\subsection{Related Work}

Our results in Section \ref{sec:asymptotic} can be viewed as generalizations of previous results by Dickerson et al. \cite{DickersonGoKa14}, who showed asymptotic existence and non-existence under a similar model but in a more limited setting where each group has only one player. In particular, these authors proved that under certain technical conditions on the probability distributions, an allocation that maximizes social welfare is envy-free with high probability if the number of items is larger than the number of players by a logarithmic factor. In fact, their result also holds when the number of players stay constant, as long as the number of items goes to infinity. Similarly, we show that a welfare-maximizing allocation is likely to be envy-free if the number of items exceeds the number of players by a logarithmic factor. While we require that the number of player per group goes to infinity, the number of groups can stay small, even constant. On the non-existence front, Dickerson et al. showed that if the utility for each item is independent and identically distributed across players, then envy-free allocations are unlikely to exist when the number of items is larger than the number of players by a linear fraction. On the other hand, our non-existence results apply to the regime where the number of items is smaller than the number of players. Note that while this regime is uninteresting in Dickerson et al.'s setting since envy-free allocations cannot exist, in our generalized setting an envy-free allocation can already exist when the number of items is at least the number of \emph{groups}.

Besides the asymptotic results on envy-free divisions, results of this type have also been shown for other fairness notions, including proportionality and the maximin share criterion. These two notions are weaker than envy-freeness when utilities are additive. Suksompong \cite{Suksompong16-2} showed that proportional allocations exist with high probability if the number of goods is a multiple of the number of players or if the number of goods grows asymptotically faster than the number of players. Kurokawa et al. \cite{KurokawaPrWa16} showed that if either the number of players or the number of items goes to infinity, then an allocation satisfying the maximin share criterion is likely to exist as long as each probability distribution has at least constant variance. Amanatidis et al. \cite{AmanatidisMaNi15} analyzed the rate of convergence for the existence of allocations satisfying the maximin share criterion when the utilities are drawn from the uniform distribution over the unit interval. Another common approach for circumventing the potential nonexistence of divisions satisfying certain fairness concepts, which we do not discuss in this paper, is by showing approximation guarantees for worst-case instances \cite{AmanatidisMaNi15,BarmanKr17,LiptonMaMo04,ProcacciaWa14,Suksompong17,Suksompong17-2}. In particular, Suksompong \cite{Suksompong17} investigated approximation guarantees for groups of agents using the maximin share criterion.

Finally, a model was recently introduced that incorporates the element of resource allocation for groups \cite{ManurangsiSu17,Suksompong16}. The model concerns the problem of finding a small ``agreeable'' subset of items, i.e., a small subset of items that a group of players simultaneously prefer to its complement. Nevertheless, in that model the preferences of only one group of players are taken into account, whereas in our work we consider the preferences of multiple groups at the same time.

\section{Preliminaries}
\label{sec:prelim}

Let a set $N$ of $n:=gn'$ players be divided into $g\geq 2$ groups $G_1,\ldots,G_g$ of $n'$ players each, and let $M:=\{1,2,\ldots,m\}$ denote the set of items. Let $g(i)$ be the index of the group containing player $i$ (i.e., player $i$ belongs to the group $G_{g(i)})$. Each item will be assigned to exactly one group, to be shared among the members of the group. We assume that each player $i\in N$ has a cardinal utility $u_i(j)$ for each item $j\in M$. We may suppose without loss of generality that $u_i(j)\in[0,1]$, since otherwise we can scale down all utilities by their maximum. We will also make a very common assumption that utilities are \emph{additive}, i.e., $u_i(M')=\sum_{j\in M'}u_i(j)$ for any player $i\in N$ and any subset of items $M'\subseteq M$. The \emph{social welfare} of an assignment is the sum of the utilities of all $n$ players from the assignment.

We are now ready to define the notion of envy-freeness. Denote the subsets of items that are assigned to the $g$ groups by $M_1,\ldots,M_g$, respectively.

\begin{definition}
Player $i$ in group $G_{g(i)}$ regards her allocation $M_{g(i)}$ as \emph{envy-free} if $u_i(M_{g(i)})\geq u_i(M_j)$ for every group $G_j\neq G_{g(i)}$. The assignment of the subsets $M_1,\ldots,M_g$ to the $g$ groups is called \emph{envy-free} if every player regards her allocation as envy-free.
\end{definition}

Next, we list two probabilistic results that will be used in our proofs. We first state the Chernoff bound, which gives us an upper bound on the probability that a sum of independent random variables is far away from its expected value.

\begin{lemma}[Chernoff bound] \label{lem:chernoff}
Let $X_1, \dots, X_r$ be independent random variables that are bounded in an interval $[0, 1]$, and let $S:=X_1 + \cdots + X_r$. We have
\begin{align*}
\Pr[S \geq (1 + \delta)\E[S]] &\leq \exp{\left(\frac{-\delta^2\E[S]}{3}\right)},
\end{align*}
and,
\begin{align*}
\Pr[S \leq (1 - \delta)\E[S]] &\leq \exp{\left(\frac{-\delta^2\E[S]}{2}\right)}
\end{align*}
for every $\delta \geq 0$.
\end{lemma}

Another lemma that we will use is the Berry-Esseen theorem. In short, it states that a sum of a sufficiently large number of independent random variables behaves similarly to a normal distribution. On the surface, this sounds like the central limit theorem. However, the Berry-Esseen theorem relies on a slightly stronger assumption and delivers a more concrete bound, which is required for our purposes.

\begin{lemma}[Berry-Esseen theorem~\cite{Berry41,Esseen42}] \label{lem:berry}
Let $X_1, \dots, X_r$ be $r$ independent and identically distributed random variables, each of which has mean $\mu$, variance $\sigma^2$, and third moment\footnote{The third moment of a random variable $X$ is defined as $\E[|X - \E[X]|^3]$.} $\rho$. Let $S:=X_1 + \cdots + X_r$. There exists an absolute constant $C_\mathit{BE}$ such that
\begin{align*}
\left|\Pr[S \leq x] - \Pr_{y \sim \mathcal{N}(\mu r, \sigma^2 r)}[y \leq x]\right| \leq \frac{\rho C_\mathit{BE}}{\sigma^3 \sqrt{r}}
\end{align*}
for every $x \in \mathbb{R}$.
Note that $\mathcal{N}(\mu r, \sigma^2 r)$ is the normal distribution with mean $\mu r$ and variance $\sigma^2 r$, i.e., its probability density function is $$f(x) = \frac{1}{\sigma \sqrt{2\pi r}}e^\frac{-(x - \mu r)^2}{2\sigma^2 r}.$$
\end{lemma}

Let us now state two assumptions on distributions of utilities; in Section~\ref{sec:asymptotic} we will work with the first and in Section~\ref{sec:existenceapprox} with the second.

\begin{itemize}
\item [\textbf{[A1]}] For each item $j\in M$, the utilities $u_i(j)\in[0,1]$ for $i\in N$ are drawn independently at random from a distribution $\mathcal{D}_j$. Each distribution $\mathcal{D}_j$ is \emph{non-atomic}, i.e., $\Pr[u_i(j)=x]=0$ for every $x\in[0,1]$. Moreover, the variances of the distributions are bounded away from zero, i.e., there exists a constant $\sigma_{min} > 0$ such that the variances of $\mathcal{D}_1, \dots, \mathcal{D}_m$ are at least $\sigma_{min}^2$.
\item [\textbf{[A2]}] For each $i\in N$ and $j\in M$, the utility $u_i(j)\in[0,1]$ is drawn independently at random from a probability distribution $\mathcal{D}_{i, j}$. The mean of each distribution is bounded away from zero, i.e., there exists a constant $\mu_{min} > 0$ such that $\E[u_i(j)] \geq \mu_{min}$ for every $i \in N, j \in M$.
\end{itemize}

Note that assumption [A2] is weaker than [A1]. Indeed, in [A2] we do not require $\mathcal{D}_{i,j}$ to be the same for every $i$. In addition, since $u_i(j)\in[0,1]$ for all $i\in N$ and $j\in M$, we have $\E[u_i(j)]\geq \E[u_i(j)^2]\geq \E[u_i(j)^2]-\E[u_i(j)]^2=\text{Var}(u_i(j)).$
Hence, the condition that the means of the distributions are bounded away from zero follows from the analogous condition on the variances.

In Section~\ref{sec:existenceapprox}, we consider the notion of {\em approximate envy-freeness}, which means that for each player, there is no allocation of another group for which the player's utility is a certain (multiplicative) factor larger than the utility of the player for the allocation of her own group. The notion is defined formally below.

\begin{definition}
We write $M_p \succsim_i^\alpha M_q$ for $\alpha \in [0, 1]$ if and only if $u_i(M_p) \geq \alpha u_i(M_q)$. Player $i$ considers an assignment $M_1, \dots, M_g$ of items to the $g$ groups \emph{$\alpha$-approximate envy-free} if $M_{g(i)} \succsim_i^\alpha M_p$ for every group $p \in \{1, \dots, g\}$. We say that an assignment is \emph{$\alpha$-approximate envy-free} if it is $\alpha$-approximate envy-free for every player $i$.
\end{definition}

Finally, we give the definition of a truthful mechanism, which we will use in Section~\ref{sec:existenceapprox}. 

\begin{definition}
A \emph{mechanism} is a function that takes as input the utility of player $i$ for item $j$ for all $i\in N$ and $j\in M$, and outputs a (possibly random) assignment of items to the groups. A mechanism is said to be \emph{truthful} if every player always obtains the highest possible (expected) utility by submitting her true utilities to the mechanism, regardless of the utilities that the remaining players submit.
\end{definition}

\section{Asymptotic Existence and Non-Existence of Fair Divisions}
\label{sec:asymptotic}

In this section, we study the existence and non-existence of fair divisions. First, we show that, when $m$ is $\Omega(n \log n)$, where $\Omega(\cdot)$ hides a sufficiently large constant, there exists an envy-free division with high probability (Theorem~\ref{thm:existence}). In particular, we prove that a welfare-maximizing allocation is likely to be envy-free. This gives rise to a simple algorithm that finds such a fair division with high probability. We also extend our existence result to the case where there are two groups but the groups need not have the same number of players; we show a similar result in this case, provided that each distribution $\mathcal{D}_j$ satisfies an additional symmetry condition (Theorem~\ref{thm:existencesymmetric}).

Moreover, on the non-existence front, we prove that when $m$ is smaller than $n$, the probability that a fair division exists is at most $1/g^{n-m}$ (Theorem~\ref{thm:nonexistence}). This has as consequences that if the number of items is less than the total number of players by a superconstant factor, or if the number of items is less than the total number of players and the number of groups is large, then the probability that an envy-free division exists is low (Corollaries~\ref{cor:nonexistence1} and \ref{cor:nonexistence2}).

We begin with our main existence result.

\begin{theorem}
\label{thm:existence}
Assume that [A1] holds. For any fixed $\sigma_{min} > 0$, there exists a constant $C > 0$ such that, for any sufficiently large $n'$, if $m > Cn\log n$, then there exists an envy-free assignment with high probability.
\end{theorem}

In fact, we not only prove that an envy-free assignment exists but also give a simple greedy algorithm that finds one such assignment with high probability. The algorithm idea is simple; we greedily assign each item to the group that maximizes the total utility of the item with respect to the players in that group. This yields an allocation that maximizes social welfare. The allocation is therefore Pareto optimal, i.e., there exists no other allocation in which every player is weakly better off and at least one player is strictly better off.
The pseudocode of the algorithm is shown below.

\begin{algorithm}
\caption{Greedy Assignment Algorithm for Multiple Groups}
\label{greedy}
\begin{algorithmic}[1]
\Procedure{Greedy\textendash Assignment\textendash Multiple}{}
\State let $M_1 = \cdots = M_g = \emptyset$.
\For{each item $j \in M$}
\State choose $k^*$ from $\argmax_{k = 1, \dots, g} \sum_{p \in G_k} u_p(j)$
\State let $M_{k^*} \leftarrow M_{k^*} \cup \{j\}$
\EndFor
\EndProcedure
\end{algorithmic}
\end{algorithm}

The analysis of the algorithm contains similarities to that of the corresponding result in the setting with one player per group \cite{DickersonGoKa14}. However, significantly more technical care will be required to handle our setting in which each group contains multiple players. This is reflected by our use of the Berry-Esseen theorem (Lemma \ref{lem:berry}). Here we provide a proof sketch that contains all the high-level ideas but leaves out some tedious details, especially calculations; the full proof can be found in the appendix.

\begin{proof}[Proof sketch of Theorem~\ref{thm:existence}]
 We will first bound $\Pr[M_{g'} \succ_i M_{g(i)}]$ for each player $i$ and each group $G_{g'} \ne G_{g(i)}$; we then use the union bound at the end to conclude Theorem~\ref{thm:existence}. To bound $\Pr[M_{g'} \succ_i M_{g(i)}]$, we define a random variable $A_{i,j}$ to be $u_i(j)$ if item $j$ is assigned to group $G_{g(i)}$ and zero otherwise. Similarly, define $B^{g'}_{i,j}$ to be $u_i(j)$ if the item is assigned to group $G_{g'}$ and zero otherwise.

 Intuitively, with respect to player $i$, $A_{i, j}$ is the utility contribution of item $j$ to the group $G_{g(i)}$. On the other hand, $B^{g'}_{i, j}$ is the utility that is ``lost'' to group $G_{g'}$. In other words, $M_{g'} \succ_i M_{g(i)}$ if and only if $S_A < S_B$, where $S_A = \sum_{j \in M} A_{i, j}$ and $S_B = \sum_{j \in M} B_{i, j}^{g'}$. We will use the Chernoff bound to estimate the probability of this event. To do so, we first need to bound $\E[A_{i, j}]$ and $\E[B_{i, j}^{g'}]$.

From symmetry between different groups, the probability that item $j$ is assigned to each group is $1/g$. Thus, we have $\E[A_{i, j}] = \frac{1}{g}\E\left[u_i(j) \mid \text{item } j \text{ is assigned to } G_{g(i)}\right]$ and $\E[B_{i, j}^{g'}] = \frac{1}{g}\E\left[u_i(j) \mid \text{item } j \text{ is assigned to } G_{g'}\right]$. It is now fairly easy to see that $\E[B_{i, j}^{g'}] \leq \mu_j/g$, where $\mu_j$ is the mean of $\mathcal{D}_j$; the reason is that the expected value of $u_i(j)$ when $j$ is not assigned to $G_{g(i)}$ is clearly at most $\mu_j$. For convenience, we will assume in this proof sketch that $\E[B_{i, j}^{g'}]$ is roughly $\mu_j/g$.

Now, we will bound the expected value of $A_{i, j}$. For each $p = 1, \dots, g$, let $X_p$ denote the sum of the utilities of item $j$ with respect to all players in $G_p$. Due to symmetry among players within the same group, we have
\begin{align*}
\E[A_{i, j}] &= \frac{1}{n'g} \E\left[X_{g(i)} \mid X_{g(i)}=\max\{X_1, \dots, X_g\}\right] \\
&= \frac{1}{n'g} \E[\max\{X_1, \dots, X_g\}].
\end{align*}
The latter equality comes from the symmetry between different groups.

Now, we use the Berry-Esseen theorem (Lemma~\ref{lem:berry}), which tells us that each of $X_1, \dots, X_g$ is close to $\mathcal{N}(\mu_j n', \Omega(\sigma_{min}^2 n'))$. With simple calculations, one can see that the expectation of the maximum of $g$ identically independent random variables sampled from $\mathcal{N}(\mu_j n', \Omega(\sigma_{min}^2 n'))$ is $\mu_j n' + \Omega(\sigma_{min}\sqrt{n'})$. Roughly speaking, we also have $$\E\left[A_{i, j}\right] = \frac{\mu_j}{g} + \Omega\left(\frac{\sigma_{min}}{g\sqrt{n'}}\right).$$

Having bounded the expectations of $A_{i, j}$ and $B_{i, j}^{g'}$, we are ready to apply the Chernoff bound. Let $\delta = \Theta\left(\frac{\sigma_{min}}{\mu_j \sqrt{n'}}\right)$ where $\Theta(\cdot)$ hides some sufficiently small constant. When $n'$ is sufficiently large, we can see that $(1 + \delta)\E[B_{i, j}^{g'}] < (1 - \delta)\E[A_{i, j}]$, which implies that $(1 + \delta)\E[S_B] < (1 - \delta)\E[S_A]$. Using the Chernoff bound (Lemma~\ref{lem:chernoff}) on $S_A$ and $S_B$, we have
\begin{align*}
\Pr[S_A \leq (1 - \delta)\E[S_A]] &\leq \exp{\left(\frac{-\delta^2\E[S_A]}{2}\right)},
\end{align*}
and,
\begin{align*}
\Pr[S_B \geq (1 + \delta)\E[S_B]] &\leq \exp{\left(\frac{-\delta^2\E[S_B]}{3}\right)}.
\end{align*}

Thus, we have
\begin{align*}
\Pr[S_A < S_B] &\leq \exp{\left(\frac{-\delta^2\E[S_A]}{2}\right)} + \exp{\left(\frac{-\delta^2\E[S_B]}{3}\right)} \\
&\leq 2\exp{\left(-\Omega\left(\frac{\sigma_{min}^2 m }{n \mu_j}\right)\right)} \\
(\text{Since } \mu_j \leq 1) &\leq 2\exp{\left(-\Omega\left(\frac{\sigma_{min}^2 m }{n}\right)\right)}.
\end{align*}

Recall that $\Pr[ M_{g'}\succ_i M_{g(i)}] = \Pr[S_A < S_B].$ Using the union bound for all $i$ and all $g' \ne g(i)$, the probability that the assignment output by the algorithm is not envy-free is at most
\begin{align*}
2n(g - 1)\exp{\left(-\Omega\left(\frac{\sigma_{min}^2 m }{n}\right)\right)},
\end{align*}
which is at most $1/m$ when $m \geq C n \log n$ for some sufficiently large $C$. This completes the proof sketch for the theorem.
\end{proof}

Unfortunately, the algorithm in Theorem \ref{thm:existence} cannot be extended to give a proof for the case where the groups do not have the same number of players. However, in a more restricted setting where there are only two groups with potentially different numbers of players and an additional symmetry condition on $\mathcal{D}_1, \dots, \mathcal{D}_m$ is enforced, a result similar to that in Theorem~\ref{thm:existence} can be shown, as stated in the theorem below.

\begin{theorem}
\label{thm:existencesymmetric}
Assume that [A1] holds. Suppose that there are only two groups but not necessarily with the same number of players; let $n_1, n_2$ denote the numbers of players of the first and second group respectively (so $n=n_1+n_2$). Assume also that $\mathcal{D}_1, \dots, \mathcal{D}_m$ are symmetric (around $1/2$)\footnote{There is nothing special about the number $1/2$; a similar result holds if the distributions are supported on a subset of an interval $[a,b]$ and are symmetric around $(a+b)/2$, for some $0<a<b$.}, i.e., \[\Pr_{X \sim \mathcal{D}_j}\left[X \leq \frac{1}{2}-x\right] = \Pr_{X \sim \mathcal{D}_j}\left[X \geq \frac{1}{2}+x\right]\] for all $x \in [0, 1/2]$. For any fixed $\sigma_{min} > 0$, there exists a constant $C > 0$ such that, for any sufficiently large $n_1$ and $n_2$, if $m > Cn\log n$, then there exists an envy-free assignment with high probability.
\end{theorem}

The algorithm is similar to that in Theorem~\ref{thm:existence}; the only difference is that, instead of assigning each item to the group with the highest \emph{total} utility over its players, we assign the item to the group with the highest \emph{average} utility, as seen in the pseudocode of Algorithm \ref{greedy2}.

\begin{algorithm}
\caption{Greedy Assignment Algorithm for Two Possibly Unequal-Sized Groups}
\label{greedy2}
\begin{algorithmic}[1]
\Procedure{Greedy\textendash Assignment\textendash Two}{}
\State let $M_1 = M_2 = \emptyset$.
\For{each item $j \in M$}
\State choose $k^*$ from $\argmax_{k = 1, 2} \frac{\sum_{p \in G_k} u_p(j)}{n_k}$
\State let $M_{k^*} \leftarrow M_{k^*} \cup \{j\}$
\EndFor
\EndProcedure
\end{algorithmic}
\end{algorithm}

The proof is essentially the same as that of Theorem~\ref{thm:existence} after the random variables defined are changed corresponding to the modification in the algorithm. For instance, $A_{i, j}$ is now defined as
\begin{align*}
A_{i, j} =
u_i(j) \cdot \mathbf{1}\left[g(i) = \argmax_{k = 1, 2} \frac{\sum_{p \in G_k} u_p(j)}{n_k}\right]
\end{align*}
where $\mathbf{1}[E]$ denotes an indicator variable for event $E$.

Due to the similarities between the two proofs, we will not repeat the whole proof. Instead, we would like to point out that all the arguments from Theorem~\ref{thm:existence} work here save for only one additional fact that we need to prove:

\begin{proposition}
Let $X_1$ and $X_2$ denote $\sum_{p \in G_1} u_p(j)/n_1$ and $\sum_{p \in G_2} u_p(j)/n_2$ respectively. Then,
\begin{align*}
\Pr[X_1 \geq X_2] = \frac{1}{2}.
\end{align*}
\end{proposition}

\begin{proof}
To show this, observe first that, since $\mathcal{D}_j$ is symmetric over $1/2$, the distributions of $X_1$ and $X_2$ are also symmetric over $1/2$. Let $f_1$ and $f_2$ be the probability density functions of $X_1$ and $X_2$ respectively, we have
\begin{align*}
\Pr[X_1 \geq X_2] &= \int_0^1 \int_0^x f_1(x)f_2(y) dy dx \\
&= \int_0^1 \int_0^x f_1(1-x)f_2(1-y) dy dx \\
&= \int_0^1 \int_x^1 f_1(x)f_2(y) dy dx \\
&= \Pr[X_2 \geq X_1].
\end{align*}

Hence, $\Pr[X_1 \geq X_2] = \Pr[X_2 \geq X_1] = 1/2$, as desired.
\end{proof}

Next, we state and prove an upper bound for the probability that an envy-free assignment exists when the number of players exceeds the number of items. Such an assignment obviously does not exist under this condition if every group contains only one player. In fact, the theorem holds even without the assumption that the variances of $\mathcal{D}_1, \dots, \mathcal{D}_m$ are at least $\sigma_{min}^2 > 0$.

\begin{theorem}
\label{thm:nonexistence}
Assume that [A1] holds. If $m<n$, then there exists an envy-free assignment with probability at most $1/g^{n-m}$.
\end{theorem}

\begin{proof}
Suppose that $m \leq n - 1$, and fix an assignment $M_1, \dots, M_g$. We will bound the probability that this assignment is envy-free. Consider any player $i$ in group $G_{g(i)}$. The probability that the assignment is envy-free for this particular player is the probability that the total utility of the player for the bundle $M_{g(i)}$ is no less than that for other bundles $M_j$. This can be written as follows:
\begin{align*}
\Pr_{u_i(1) \in \mathcal{D}_1, \dots, u_i(m) \in \mathcal{D}_m}\left[\sum_{l \in M_{g(i)}} u_i(l) = \max_{k = 1, \dots, g}{\sum_{l \in M_k} u_i(l)}\right].
\end{align*}

For each $j = 1, \dots, g$, define $p_j$ as
\begin{align*}
p_j = \Pr_{x_1 \in \mathcal{D}_1, \dots, x_m \in \mathcal{D}_m}\left[\sum_{l \in M_j} x_l = \max_{k = 1, \dots, g}{\sum_{l \in M_k} x_l}\right].
\end{align*}
Notice that the probability that the assignment is envy-free for player $i$ is $p_{g(i)}$.

Since $u_i(1), \dots, u_i(m)$ is chosen independently of $u_{i'}(1), \dots, u_{i'}(m)$, for every $i' \ne i$, the probability that this assignment is envy-free for every player is simply the product of the probability that the assignment is envy-free for each player, i.e.,
\begin{align*}
\prod_{i = 1}^{n} p_{g(i)} &= \prod_{j=1}^{g} p_{j}^{n'}.
\end{align*}

Using the inequality of arithmetic and geometric means, we arrive at the following bound:
\begin{align*}
\prod_{j=1}^{g} p_{j}^{n'} \leq \left(\frac{1}{g} \sum_{j=1}^{g} p_j\right)^{n'g}.
\end{align*}

Recall our assumption that the distributions $\mathcal{D}_j$ are non-atomic. Hence we may assume that the events $\sum_{l \in M_j} x_l = \max_{k = 1, \dots, g}{\sum_{l \in M_k} x_l}$ are disjoint for different $j$. This implies that $\sum_{j=1}^g p_j = 1$. Thus, the probability that this fixed assignment is envy-free is at most
\begin{align*}
\left(\frac{1}{g} \sum_{j=1}^{g} p_j\right)^{n'g} = \left(\frac{1}{g}\right)^{n'g} = \frac{1}{g^n}.
\end{align*}

Finally, since each assignment is envy-free with probability at most $1/g^{n}$ and there are $g^m$ possible assignments, by union bound the probability that there exists an envy-free assignment is at most $1/g^{n-m}$. This completes the proof of the theorem.
\end{proof}

The following corollaries can be immediately  derived from Theorem \ref{thm:nonexistence}. They say that an envy-free allocation is unlikely to exist when the number of items is less than the number of players by a superconstant factor, or when the number of items is less than the number of players and the number of groups is large.

\begin{corollary}
\label{cor:nonexistence1}
Assume that [A1] holds. When $m=n-\omega(1)$, the probability that there exists an envy-free assignment converges to zero as $n\rightarrow\infty$.
\end{corollary}

\begin{corollary}
\label{cor:nonexistence2}
Assume that [A1] holds. When $m<n$, the probability that there exists an envy-free assignment converges to zero as $g\rightarrow\infty$.
\end{corollary}




\section{Truthful Mechanism for Approximate Envy-Freeness}
\label{sec:existenceapprox}

While the algorithms in Section \ref{sec:asymptotic} translate to mechanisms that yield with high probability envy-free divisions that are compatible with social welfare assuming that players are truth-telling, the resulting mechanisms suffer from the setback that they are easily manipulable. Indeed, since they aim to maximize (total or average) welfare, strategic players will declare their values for the items to be high, regardless of what the actual values are. This presents a significant disadvantage: Implementing these mechanisms in most practical situations, where we do not know the true valuations of the players and have no reason to assume that they will reveal their valuations in a honest manner, can lead to potentially undesirable outcomes.

In this section, we work with the weaker notion of approximate envy-freeness and show that a simple truthful mechanism yields an approximately envy-free assignment with high probability. In particular, we prove that the random assignment mechanism, which assigns each item to a player chosen uniformly and independently at random, is likely to produce such an assignment. In the setting where each group consists of only one player, Amanatidis et al. \cite{AmanatidisBM16} showed that when the distribution is as above and the number of items $m$ is large enough compared to $n$, the random assignment mechanism yields an approximately envy-free assignment with high probability. Our statement is an analogous statement for the case where each group can have multiple players.

\begin{theorem} \label{thm:existenceapprox}
Assume that [A2] holds. For every $\alpha \in [0, 1)$, there exists a constant $C$ depending only on $\alpha$ and $\mu_{min}$ such that, if $m > C g \log n$, then the random assignment, where each item $j \in M$ is assigned independently and uniformly at random to a group, is $\alpha$-approximate envy-free with high probability.
\end{theorem}

Before we prove Theorem~\ref{thm:existenceapprox}, we note some ways in which our result is stronger than that of Amanatidis et al.'s apart from the fact that multiple players per group are allowed in our setting. First, Amanatidis et al. required $\mathcal{D}_{i, j}$ to be the same for all $j$, which we do not assume here. Next, they only showed that the random assignment is likely to be \emph{approximately proportional}, a weaker notion that is implied by approximate envy-freeness. Moreover, in their result, $m$ needs to be as high as $\Omega(n^2)$, whereas in our case it suffices for $m$ to be in the range $\Omega(g \log n)$. Finally, we also derive a stronger probabilistic bound; they showed a ``success probability'' of the algorithm of $1 - O(n^2/m)$, while our success probability is $1 - \exp(-\Omega(m/g))$.

\begin{proof}[Proof of Theorem~\ref{thm:existenceapprox}]
For each player $i \in N$, each item $j \in M$ and each $p \in \{1, \dots, g\}$, let $A_{i, j}^p$ be a random variable representing the contribution of item $j$'s utility with respect to player $i$ to group $G_p$, i.e., $A_{i, j}^p$ is $u_i(j)$ if item $j$ is assigned to group $G_p$ and is zero otherwise.

Define $S_i^p:=\sum_{j \in M} A_{i, j}^p$. Observe that each player $i$ considers the assignment to be $\alpha$-approximate envy-free if and only if $S_i^{g(i)} \geq \alpha S_i^p$ for every $p$. Let $\delta = \frac{1 - \alpha}{1 + \alpha}$; from this choice of $\delta$ and since $\E[S_i^p]$ is equal for every $p$, we can conclude that $S_i^{g(i)} \geq \alpha S_i^p$ is implied by $S_i^{g(i)} \geq (1 - \delta) \E[S_i^{g(i)}]$ and $S_i^p \leq (1 + \delta)\E[S_i^p]$.
In other words, we can bound the probability that the random assignment is not $\alpha$-approximate envy-free as follows.
\begin{align*}
&\Pr[\exists i \in N, p \in \{1, \dots, g\}: S_i^{g(i)} < \alpha S_i^p] \\
&\phantom{{}=6}\leq \sum_{i \in N, p \in \{1, \dots, g\}} \Pr[S_i^{g(i)} < \alpha S_i^p] \\
&\phantom{{}=6}\leq \sum_{i \in N, p \in \{1, \dots, g\}} \Pr[S_i^{g(i)} < (1 - \delta) \E[S_i^{g(i)}]  \vee S_i^p > (1 + \delta)\E[S_i^{g(i)}]] \\
&\phantom{{}=6}\leq \sum_{i \in N, p \in \{1, \dots, g\}} (\Pr[S_i^{g(i)} < (1 - \delta) \E[S_i^{g(i)}]] + \Pr[S_i^p > (1 + \delta)\E[S_i^{p}]]). 
\end{align*}

Since $S_i^p = \sum_{j \in M} A_{i, j}^p$ and $A_{i, j}^p$'s are independent and lie in $[0, 1]$, we can use Chernoff bound (Lemma \ref{lem:chernoff}) to upper bound the last terms. Hence, the probability that the alloaction is not $\alpha$-approximate envy-free is at most
\begin{align*}
\sum_{i \in N, p \in \{1, \cdots, g\}} \exp\left(\frac{-\delta^2 \E[S_i^{g(i)}]]}{2}\right) + \exp\left(\frac{-\delta^2 \E[S_i^p]]}{3}\right).
\end{align*}

Finally, observe that \[\E[S_i^p] = \sum_{j \in M} \E[A_{i, j}^p] = \sum_{j \in M} \frac{1}{g} \E[u_i(j)] \geq \frac{m \mu_{min}}{g}.\] This means that the desired probability is bounded above by
\begin{align*}
&\sum_{i \in N, p \in \{1, \dots, g\}} \exp\left(\frac{-\delta^2 m \mu_{min}}{2g}\right) + \exp\left(\frac{-\delta^2 m \mu_{min}}{3g}\right) \\
&\phantom{{}=6}\leq 2ng \exp\left(\frac{-\delta^2 m \mu_{min}}{3g}\right) \\
&\phantom{{}=6}\leq \exp\left(-\frac{\delta^2 m \mu_{min}}{3g} + 3\log n\right).
\end{align*}

When $m > \left(\frac{10}{\mu_{min}\delta^2}\right)g \log n$, the above expression is at most $\exp(-\Omega(m/g))$, concluding our proof.
\end{proof}

\section{Concluding Remarks}

In this paper, we study a generalized setting for fair division that allows interested parties to contain multiple players, possibly with highly differing preferences. This setting allows us to model several real-world cases of fair division that cannot be done under the traditional setting. We establish almost-tight bounds on the number of players and items under which a fair division is likely or unlikely to exist. Furthermore, we consider the issue of truthfulness and show that a simple truthful mechanism produces an assignment that is approximately envy-free with high probability.

While the assumptions of additivity and independence are somewhat restrictive and might not apply fully to settings in the real world, our results give indications as to what we can expect if the assumptions are relaxed, such as if a certain degree of dependence is introduced. An interesting future direction is to generalize the results to settings with more general valuations. In particular, if the utility functions are low-degree polynomials, then one could try applying the invariance principle \cite{MosselOdOl10}, which is a generalization of the Berry-Esseen theorem that we use. 

We end the paper with some questions that remain after this work. A natural question is whether we can generalize our existence and non-existence results (Theorems \ref{thm:existence} and \ref{thm:nonexistence}) to the setting where the groups do not contain the same number of players. This non-symmetry between the groups seems to complicate the approaches that we use in this paper. For example, it breaks the greedy algorithm used in Theorem \ref{thm:existence}. Nevertheless, it might still be possible to prove existence of an envy-free division using other algorithms or without relying on a specific algorithm.

Another direction for future research is to invent procedures for computing envy-free divisions, whenever such divisions exist, for the general setting where each group contains multiple players and players have arbitrary (not necessarily additive) preferences. Even procedures that only depend on rankings of single items \cite{BramsKiKl14} do not appear to extend easily to this setting. Indeed, if a group contains two players whose preferences are opposite of each other, it is not immediately clear what we should assign to the group. It would be useful to have a procedure that produces a desirable outcome, even for a small number of players in each group.

Lastly, one could explore the limitations that arise when we impose the condition of truthfulness, an important property when we implement the mechanisms in practice. For instance, truthful allocation mechanisms have recently been characterized in the case of two players \cite{AmanatidisBiCh17}, and it has been shown that there is a separation between truthful and non-truthful mechanisms for approximating maximin shares \cite{AmanatidisBM16}. In our setting, a negative result on the existence of a truthful mechanism that yields an envy-free division with high probability would provide such a separation as well, while a positive result in this direction would have even more consequences for practical applications.

\section*{Acknowledgments}

The authors thank the anonymous reviewers for their helpful feedback. Warut Suksompong acknowledges support from a Stanford Graduate Fellowship.

\bibliographystyle{abbrv}
\bibliography{main}

\appendix

\section{Appendix}

\subsection{Proof of Theorem \ref{thm:existence}}

First we list the following well-known fact, which allows us to easily determine the mean of a random variable from its cumulative density function.

\begin{proposition}
\label{prop:meanbymass}
Let $X$ be a non-negative random variable. Then
\begin{align*}
\E[X] = \int_0^\infty \Pr[X \geq x] dx.
\end{align*}
\end{proposition}

To analyze the algorithm, consider any player $i$ and any group $g' \ne g(i)$. We will first bound the probability that $M_{g'} \succ_i M_{g(i)}$. To do this, for each item $j \in M$, define $A_{i, j}$ as
\begin{align*}
A_{i, j} =
u_i(j)\cdot\textbf{1}\left[g(i) = \argmax_{k = 1, \dots, g} \sum_{p \in G_k} u_p(j)\right],
\end{align*}
where $1\left[g(i) = \argmax_{k = 1, \dots, g} \sum_{p \in G_k} u_p(j)\right]$ is an indicator random variable that indicates whether $g(i) = \argmax_{k = 1, \dots, g} \sum_{p \in G_k} u_p(j)$.
Similarly, define $B_{i, j}^{g'}$ as
\begin{align*}
B_{i, j}^{g'} =
u_i(j)\cdot\textbf{1}\left[g' = \argmax_{k = 1, \dots, g} \sum_{p \in G_k} u_p(j)\right].
\end{align*}
Moreover, suppose that $\mathcal{D}_j$ has mean $\mu_j$ and variance $\sigma_j^2$.

Notice that, with respect to player $i$, $A_{i, j}$ is the utility that item $j$ contributes to $g(i)$ whereas $B^{g'}_{i, j}$ is the utility that item $j$ contributes to $g'$. In other words, $M_{g'} \succ_i M_{g(i)}$ if and only if $\sum_{j \in M} A_{i, j} < \sum_{j \in M} B_{i, j}^{g'}$. To bound $\Pr\left[M_{g'} \succ_i M_{g(i)}\right]$, we will first bound $\E\left[A_{i, j}\right]$ and $\E\left[B_{i, j}^{g'}\right]$. Then, we will use the Chernoff bound to bound $\Pr\left[\sum_{j \in M} A_{i, j} < \sum_{j \in M} B_{i, j}^{g'}\right]$.

Observe that, due to symmetry, we can conclude that
\begin{align*}
\E\left[u_i(j)\cdot\textbf{1}\left[g' = \argmax_{k = 1, \dots, g} \sum_{p \in G_k} u_p(j)\right]\right] = \E\left[u_i(j)\cdot\textbf{1}\left[g'' = \argmax_{k = 1, \dots, g} \sum_{p \in G_k} u_p(j)\right]\right]
\end{align*}
for any $g'' \ne g(i)$. Thus, we can now rearrange $B_{i, j}^{g'}$ as follows:
\begin{align*}
\E\left[B_{i, j}^{g'}\right] &= \frac{1}{g - 1} \left(\sum_{g'' \ne g(i)} \E\left[u_i(j)\cdot\textbf{1}\left[g'' = \argmax_{k = 1, \dots, g} \sum_{p \in G_k} u_p(j)\right]\right]\right) \\
&= \frac{1}{g - 1} \left(\E\left[u_i(j) \sum_{g'' \ne g(i)}\textbf{1}\left[g'' = \argmax_{k = 1, \dots, g} \sum_{p \in G_k} u_p(j)\right]\right]\right) \\
&= \frac{1}{g - 1} \left(\E\left[u_i(j) \left(1 - \textbf{1}\left[g(i) = \argmax_{k = 1, \dots, g} \sum_{p \in G_k} u_p(j)\right]\right)\right]\right).
\end{align*}

Hence, we have
\begin{align}
\label{eq:AtoB}
\E\left[B_{i, j}^{g'}\right] = \frac{1}{g - 1} \left(\mu_j - \E\left[A_{i, j}\right]\right).
\end{align}

Now, consider $A_{i, j}$. Again, due to symmetry, we have
\begin{align*}
\E\left[A_{i, j}\right] &= \frac{1}{n'} \left(\sum_{i' \in G_{g(i)}} \E\left[u_{i'}(j)\cdot\textbf{1}\left[g(i) = \argmax_{k = 1, \dots, g} \sum_{p \in G_k} u_p(j)\right]\right]\right) \\
&= \frac{1}{n'} \E\left[\left(\sum_{i' \in G_{g(i)}} u_{i'}(j)\right)\cdot\textbf{1}\left[g(i) = \argmax_{k = 1, \dots, g} \sum_{p \in G_k} u_p(j)\right]\right].
\end{align*}

Let $\mathcal{S}$ denote the distribution of the sum of $n'$ independent random variables, each drawn from $\mathcal{D}_j$. It is obvious that $\sum_{p \in G_k} u_p(j)$ is drawn from $\mathcal{S}$ independently for each $k$. In other words, $\E\left[A_{i, j}\right]$ can be written as
\begin{align*}
\E\left[A_{i, j}\right] &= \frac{1}{n'} \E\left[X_1 \cdot\textbf{1}\left[X_1 = \max\{X_1, \dots, X_g\}\right]\right].
\end{align*}
The expectation on the right is taken over $X_1, \dots, X_g$ sampled independently from $\mathcal{S}$.

From symmetry among $X_1, \dots, X_g$, we can further derive the following:
\begin{align*}
\E\left[A_{i, j}\right] &= \frac{1}{n'} \Pr\left[X_1 = \max\{X_1, \dots, X_g\}\right] \E\left[X_1 \mid X_1 = \max\{X_1, \dots, X_g\}\right] \\
&= \frac{1}{n'g} \E\left[X_1 \mid X_1 = \max\{X_1, \dots, X_g\}\right] \\
&= \frac{1}{n'g} \E\left[\max\{X_1, \dots, X_g\}\right].
\end{align*}

Consider the distribution of $\max\{X_1, \dots, X_g\}$. Let us call this distribution $\mathcal{Y}$. Notice that $\E\left[\max\{X_1, \dots, X_g\}\right]$ is just the mean of $\mathcal{Y}$, i.e.,
\begin{align}
\label{eq:meanA}
\E\left[A_{i, j}\right] = \frac{1}{n' g} \E_{Y \sim \mathcal{Y}}[Y].
\end{align}

To bound this, let $F_S$ and $F_Y$ be the cumulative density functions of $\mathcal{S}$ and $\mathcal{Y}$ respectively. Notice that $F_Y(x) = F_S(x)^g$ for all $x$. Applying Proposition~\ref{prop:meanbymass} to $\mathcal{S}$ and $\mathcal{Y}$ yields the following:
\begin{align*}
\E_{S \sim \mathcal{S}}[S] = \int_0^\infty (1 - F_S(x))dx,
\end{align*}
and,
\begin{align*}
\E_{Y \sim \mathcal{Y}}[Y] = \int_0^\infty (1 - F_S(x)^g)dx.
\end{align*}

By taking the difference of the two, we have
\begin{align*}
\E_{Y \sim \mathcal{Y}}[Y]  = \E_{S \sim \mathcal{S}}[S] + \int_0^\infty F_S(x)\left(1 - F_S(x)^{g - 1}\right)dx.
\end{align*}

To bound the right hand side, recall that $\mathcal{S}$ is just the distribution of the sum of $n'$ independent random variables sampled according to $\mathcal{D}_j$. Note that the third moment of $\mathcal{D}_j$ is at most 1 because it is bounded in $[0, 1]$. Thus, by applying the Berry-Esseen Theorem (Lemma 2), we have
\begin{align*}
\left|F_S(x) - \Pr_{y \sim \mathcal{N}(\mu_j n', \sigma_j^2 n')}[y \leq x]\right| \leq \frac{C_{BE}}{\sigma_j^3 \sqrt{n'}}.
\end{align*}
for all $x \in \mathbb{R}$. When $n'$ is sufficiently large, the right hand side is at most 0.1.

Moreover, it is easy to check that $\Pr_{y \sim \mathcal{N}(\mu_j n', \sigma_j^2 n')}[y \leq x] \in [0.5, 0.85]$
for every $x \in \left[\mu_j n', \mu_j n' + \sigma_j \sqrt{n'}\right]$. Hence, $F_S(x) \in [0.4, 0.95]$ for every $x \in \left[\mu_j n', \mu_j n' + \sigma_j \sqrt{n'}\right]$.

Now, we can bound $\E_{Y \sim \mathcal{Y}}[Y]$ as follows:
\begin{align*}
\E_{Y \sim \mathcal{Y}}[Y]  &= \E_{S \sim \mathcal{S}}[S] + \int_0^\infty F_S(x)\left(1 - F_S(x)^{g - 1}\right)dx \\
&= \mu_j n' + \int_0^\infty F_S(x)\left(1 - F_S(x)^{g - 1}\right)dx \\
(\text{Since } F_S(x)\left(1 - F_S(x)^{g - 1}\right) \geq 0) &\geq \mu_j n' + \int_{\mu_j n'}^{\mu_j n' + \sigma_j \sqrt{n'}} F_S(x)\left(1 - F_S(x)^{g - 1}\right)dx \\
&\geq \mu_j n' + \int_{\mu_j n'}^{\mu_j n' + \sigma_j \sqrt{n'}} (0.4)(0.05)dx \\
&= \mu_j n' + \sigma_j \sqrt{n'}/50 \\
(\text{Since } \sigma_j \geq \sigma_{min}) &\geq \mu_j n' + \sigma_{min} \sqrt{n'}/50.
\end{align*}

Plugging the above inequality into equation~(\ref{eq:meanA}), we can conclude that
\begin{align*}
\E\left[A_{i, j}\right] = \frac{1}{n'g} \E_{Y \in \mathcal{Y}}[Y] \geq \frac{\mu_j}{g} + \frac{\sigma_{min}}{50 g \sqrt{n'}}.
\end{align*}

From this and equation~(\ref{eq:AtoB}), we have
\begin{align*}
\E\left[B_{i, j}^{g'}\right] = \frac{1}{g - 1} \left(\mu_j - \E\left[A_{i, j}\right]\right)
\leq \frac{1}{g - 1} \left(\mu_j - \frac{\mu_j}{g}\right)
= \frac{\mu_j}{g}.
\end{align*}

Now, define $C_{i, j}^{g'}$ as $C_{i, j}^{g'} = B_{i, j}^{g'} + \left(\mu_j/g - \E\left[B_{i, j}^{g'}\right]\right)$. Notice $\E\left[C_{i, j}^{g'}\right] = \mu_j/g$.

As stated earlier, $M_{g'} \succ_i M_{g(i)}$ if and only if $\sum_{j \in M} A_{i, j} < \sum_{j \in M} B_{i, j}^{g'}$. Let $S_A = \sum_{j \in M} A_{i, j}, S_B = \sum_{j \in M} B_{i, j}^{g'}, S_C = \sum_{j \in M} C_{i, j}^{g'}$ and let $\delta = \frac{\sigma_{min}}{200\mu_j \sqrt{n'}}$. Notice that, since we assume that the variance of $\mathcal{D}_j$ is positive, $\mu_j$ is also non-zero, which means that $\delta$ is well-defined. Using Chernoff bound (Lemma 1) on $S_A$ and $S_C$, we have
\begin{align*}
\Pr[S_A \leq (1 - \delta)\E[S_A]] &\leq \exp{\left(\frac{-\delta^2\E[S_A]}{2}\right)},
\end{align*}
and,
\begin{align*}
\Pr[S_C \geq (1 + \delta)\E[S_C]] &\leq \exp{\left(\frac{-\delta^2\E[S_C]}{3}\right)}.
\end{align*}

Moreover, when $n'$ is large enough, we have $(1 - \delta)\E[S_A] \geq (1 + \delta)\E[S_C]$. Thus, we have
\begin{align*}
\Pr[S_A < S_C] &\leq \exp{\left(\frac{-\delta^2\E[S_A]}{2}\right)} + \exp{\left(\frac{-\delta^2\E[S_C]}{3}\right)} \\
&\leq \exp{\left(\frac{-\delta^2 m \mu_j}{2g}\right)} + \exp{\left(\frac{-\delta^2m \mu_j}{3g}\right)} \\
&\leq 2\exp{\left(\frac{-\sigma_{min}^2 m }{120000g n' \mu_j}\right)} \\
(\text{Since } \mu_j \leq 1) &\leq 2\exp{\left(\frac{-\sigma_{min}^2 m }{120000n}\right)}.
\end{align*}

Due to how $C_{i, j}^{g'}$ is defined, we have $\Pr[S_A < S_C] \geq \Pr[S_A < S_B] = \Pr[ M_{g'}\succ_i M_{g(i)}].$ Using the union bound for all $i$ and all $g' \ne g(i)$, the probability that the assignment output by the algorithm is not envy-free is at most
\begin{align*}
2n(g - 1)\exp{\left(\frac{-\sigma_{min}^2 m }{120000n}\right)},
\end{align*}
which is at most $1/m$ when $m \geq C n \log n$ for some sufficiently large $C$. This completes the proof for the theorem.

\end{document}